\documentclass{article}
\usepackage{textcomp}
\usepackage{amsthm}
\usepackage{authblk}
\usepackage{amsfonts}
\usepackage{amssymb,bm}
\usepackage{mathrsfs}
\usepackage{mathtools}
\usepackage{amsmath}
\usepackage{enumitem}
\usepackage[colorlinks=true, linkcolor=blue, citecolor=blue, urlcolor=blue]{hyperref}
\usepackage{cite}
\usepackage{balance}

\newtheorem{theorem}{Theorem}
\newtheorem{definition}{Definition}

\newtheorem{proposition}[theorem]{Proposition}

\newtheorem{remark}{Remark}

\def\BibTeX{{\rm B\kern-.05em{\sc i\kern-.025em b}\kern-.08em
    T\kern-.1667em\lower.7ex\hbox{E}\kern-.125emX}}

\begin{document}
\title{Multi-Target Observability}
\author[1]{Debadrita Banerjee}
\author[2]{Debjani Mitra}
\author[3]{Rajesh Dey}
\author[4]{Mudassir Khan}
\author[5]{Lalan Kumar}

\affil[1]{Electrical Engineering, IIT Delhi, India}
\affil[2]{Electronics Engineering, IIT ISM Dhanbad}
\affil[3]{Information \& Technology GNUB, India}
\affil[4]{Computer Science, ACTKK University, Saudi Arabia}
\affil[5]{Electrical Engineering, IIT Delhi, India}
\date{}

\maketitle

\begin{abstract}
In this paper, we mainly focus on the problem of multi-target observability, focusing on the unique state estimation criteria for multiple targets. We derive the condition which is necessary as well as sufficient for observability using bearing angles with multiple higher-order dynamics observed by a single observer. We then establish an alternative notion of observability by analyzing ambiguous target trajectories and deriving the condition which is NECNDSUF (Nec. and Suff.) for multi-target observability, considering three types of measurements: Doppler-only, bearing-only, and combined Doppler and bearing measurements, which offers insights that can improve target distinguishability, trajectory reconstruction, and overall tracking accuracy. 

\end{abstract}

\section{Introduction}
Observability plays a crucial role in Target Motion Analysis. Depending on the relative positions of the targets and observers, there are certain relative geometry of the target and the observing platform that can lead to situations where a target becomes unobservable by the observer and couldn't produce a distinct tracking solution. In such scenarios, the observer may not be able to accurately determine the target's position or trajectory. A bearings-only tracking system does not always yield a unique tracking solution \cite{shar1999practical}.
Ensuring that the geometric arrangement avoids these problematic configurations is essential for maintaining accurate and reliable target tracking. The observability of a single target has been studied in the literature for BOTMA \cite{shar1999practical,jauffret1996observability,nardone1981observability,le1997discrete,fogel1988nth,song1996observability,li2010observability,jiang2020observability} , DOTMA and DBTMA \cite{becker1996general,jauffret1996observability,shames2013doppler,torney2007localization,xiao2010observability}, but criteria for the observability of multiple targets have not yet been derived. BOTMA requires effective maneuvers by the observer to achieve observability, making it more challenging. The accuracy of BOTMA highly depends on these maneuvers. Some studies suggest solutions for optimizing these maneuvers \cite{le1997discrete,passerieux1998optimal,le1999optimizing}, while others highlight maneuvers that keep the source's trajectory unobservable \cite{jauffret1996observability,nardone1981observability}. In \cite{jauffret1996observability}  various conditions for observability are discussed, including the necessity and effectiveness of observer maneuvers, while
\cite{nardone1981observability}details NECNDSUF conditions for such maneuvers to enhance target observability.
Additionally, optimal observer maneuvers are further explored in \cite{le1997discrete}, providing a comprehensive analysis of strategies to improve target observability. Criteria for observing a target's Nth-order dynamics via direction measurements are provided in \cite{fogel1988nth} with simple derivations that avoid using an observability matrix or nonlinear equations, showing that prior first-order conditions are necessary but insufficient. 

In system theory, observability is clearly defined within linear frameworks, but when it comes to nonlinear systems, there are multiple interpretations of observability. Commonly employed tools rooted in Lie Algebra often prove to be cumbersome for such cases \cite{hermann1977nonlinear}. Typically, when dealing with measurement nonlinearities, analysis methods like the "ratio test" or the "strongly positive semidefinite condition" are commonly employed \cite{kou1973observability,fujisawa1971some}. Though the observability issues in BOTMA are nonlinear and complex \cite{griffith1971observability,hermann1977nonlinear}, the pseudolinearized version simplifies them to a linear observability problem which is addressed in \cite{song1996observability,nardone1981observability}.
In \cite{nardone1981observability} the system is reformulated into a linear equivalent, enabling a straightforward application of observability test. Following this, analysis reveals a third-order nonlinear differential equation that encompasses the relevant limitations on the motion of own ship. After solving this equation it has been demonstrated that specific maneuvers are deemed unsuitable, even when bearing rates are non-zero. 
The study in \cite{hammel1985observability} thoroughly investigates the observability traits of the estimation algorithms in three dimensions, capable of handling various combinations of azimuth, conical, and elevation angle measurements simultaneously by converting nonlinearities into linearities.
Again Deriving both NECNDSUF conditions for solving the observability problem in general three-dimensional Target Motion Analysis (TMA) based on angle measurements appeared excessively challenging using linear equations. However,in \cite{becker1993simple} it is demonstrated that this derivation becomes notably simpler when employing a suitable observability criterion.
The work in \cite{jiang2020observability} addresses the observability problem by proposing metrics based on the condition number to quantify the degree of observability and assess tracking performance, going beyond the qualitative results of the traditional rank criterion of nonlinear observability problems. These metrics can also be used to optimize sensor trajectories and configurations, crucial for enhancing tracking performance.
In complex network systems, key goals include identifying sufficient sensor nodes for observability and selecting the best sensor combination for accurate state estimation.\cite{montanari2020observability} critically reviews observability approaches, highlighting graph-theoretical methods, their limitations, and applications in power grids and systems with multiple agents. Numerous studies in the literature\cite{lu2017observability,liu2013observability,letellier2018nonlinear,10.1007/978-981-97-0154-4_44,haber2017state} have adopted a graph-theoretical perspective to address the observability of network systems.\\
Section~\ref{sec:sec1} outlines the background of observability for multiple targets.
Section~\ref{sec:sec2} outlines the NECNDSUF conditions for observability of multiple targets when only bearing measurements are available and Section~\ref{sec:sec3} outeline the concept of observability through trajectory ambiguity when only Doppler frequency measurements are available and when both bearing and Doppler frequency measurements are available. 

\section{Theory on Observability}\label{sec:sec1}
There are various definitions of observability in the literature \cite{jauffret1996observability,nardone1981observability,becker1996general,montanari2020observability}, such as those based on the observability matrix or observability Grammians. In the upcoming section, we will focus on two specific definitions (Definition \ref{def:def1} and Definition \ref{def:def2}) to derive the observability conditions for a linear time varying multiple-target system.\\
Consider a noise-free, linear time-varying system with $n$-dimensional state. Let the set of $M$ targets be defined as $\mathcal{M}= \left\{1,2,...M\right\}$ where the states of all targets are combined into a single super state vector $\mathbf{x}(t)$ can be written as:
\begin{align}
   \dot{\mathbf{x}}(t) = \mathbf{E}(t)\mathbf{x}(t) + \mathbf{F}(t)\mathbf{u}(t), \quad \label{eqn. 1}\\ 
   \mathbf{z}(t) = \mathbf{G}(t)\mathbf{x}(t), \quad  t \in [t_{i}, t_{f}], \label{eqn. 2}
\end{align}
where, \( \mathbf{x}(t) \in \mathbb{R}^{nM \times 1} \),
\( \mathbf{E}(t) \in \mathbb{R}^{ nM \times nM} \) is the state transition matrix,
\( \mathbf{F}(t) \in \mathbb{R}^{nM \times p} \) is the input matrix,
\( \mathbf{u}(t) \in \mathbb{R}^{ p \times 1} \) is the input vector,
\( \mathbf{G}(t) \in \mathbb{R}^{ M \times nM }\) is the measurement matrix, and
\( \mathbf{z}(t) \in \mathbb{R}^{ M \times 1} \) is the measurement.
The solution to the differential equation \eqref{eqn. 1} is given by (see Appendix~\ref{observability proof}) \cite{jauffret1996observability}:
\begin{align}\label{eqn. 3}
    \mathbf{x}(t) = \tilde{\mathbf{\Phi}}(t,t_{i})\mathbf{x}(t_{i}) + \tilde{\mathbf{R}}(t),
\end{align}
where the \textbf{state transition matrix} \( \tilde{\mathbf{\Phi}}(t,t_{i}) \) is the unique matrix satisfying:
\[
\frac{d}{dt}\tilde{\mathbf{\Phi}}(t,t_{i}) = \mathbf{E}(t)\tilde{\mathbf{\Phi}}(t,t_{i}) \quad \text{with} \quad \tilde{\mathbf{\Phi}}(t_{i},t_{i}) = \mathbf{I}_{nM \times nM},
\]
and $\tilde{\mathbf{R}}(t) = \int_{t_{i}}^{t} \tilde{\mathbf{\Phi}}(t,\tau)\mathbf{F}(\tau)\mathbf{U}(\tau) d\tau,$ which is an \( nM \times 1 \) vector.
Combining \eqref{eqn. 1} and \eqref{eqn. 2}, we obtain:
\begin{align} \label{eqn.5}
    \mathbf{G}(t)\tilde{\mathbf{\Phi}}(t,t_{i})\mathbf{x}(t_{i}) = \mathbf{z}(t) - \mathbf{G}(t)\tilde{\mathbf{R}}(t),
\end{align}
which implies
\begin{align}
    \left[ \int_{t_{i}}^{t_{f}} \tilde{\mathbf{\Phi}}^\top(t, t_{i}) \mathbf{G}^\top(t) \mathbf{G}(t) \tilde{\mathbf{\Phi}}(t, t_{i}) \, dt \right] \mathbf{x}(t_{i}) = \int_{t_{i}}^{t_{f}} \tilde{\mathbf{\Phi}}^\top(t, t_{i})\nonumber \\ \mathbf{G}^\top(t)  \left[ \mathbf{z}(t) - \mathbf{G}(t) \tilde{\mathbf{R}}(t) \right] \, dt.
\end{align}
In order to solve for the unknown \( \mathbf{x}(t_{i}) \), the term $\int_{t_{i}}^{t_{f}} \tilde{\mathbf{\Phi}}^\top(t, t_{i}) \mathbf{G}^\top(t) \mathbf{G}(t) \tilde{\mathbf{\Phi}}(t, t_{i}) \, dt $ must be invertible. Invertibility ensures that the system of equations is solvable and that we can uniquely determine \( \mathbf{x}(t_{i}) \).

Thus the invertibility of $\int_{t_{i}}^{t_{f}} \tilde{\mathbf{\Phi}}^\top(t, t_{i}) \mathbf{G}^\top(t) \mathbf{G}(t) \tilde{\mathbf{\Phi}}(t, t_{i}) \, dt,$ is ensured if, for any $\mathbf{y} \neq 0$,\\ $\mathbf{y}^T \left[ \int_{t_{i}}^{t_{f}} \tilde{\mathbf{\Phi}}^\top(t, t_{i}) \mathbf{G}^\top(t) \mathbf{G}(t) \tilde{\mathbf{\Phi}}(t, t_{i}) \, dt \right] \mathbf{y} \neq 0.$
Alternatively,
\begin{align}
    \text{For any } \mathbf{y} \neq 0, \quad \int_{t_{i}}^{t_{f}} \| \mathbf{G}(t) \tilde{\mathbf{\Phi}}(t, t_{i}) \mathbf{y}\|^2 \, dt > 0.
\end{align}
In other words, to ensure invertibility, for any non-zero $ \mathbf{y}$, there must exist some  $t \in [t_{i}, t_{f}]$ such that $\| \mathbf{G}(t) \tilde{\mathbf{\Phi}}(t, t_{i}) \mathbf{y} \| \neq 0.$ (where $||.||$ usually means $2$-Norm)
\begin{definition}[Observability for Multiple Targets in Linear system] \label{def:def1}
    The system defined and described by equations (1) and (2) is observable over the interval  \([t_{i}, t_{f}]\) iff for any vector \(\mathbf{y} \in \mathbb{R}^{nM \times 1}\), there exists \(t \in [t_{i}, t_{f}]\) such that
\begin{align} \label{eqn 10}
    \|\mathbf{G}(t)\tilde{\mathbf{\Phi}}(t, t_{i})\mathbf{y}\| \neq 0.
\end{align}
\end{definition}
Another way to define observability\cite{10101861,becker1996general} in an arbitrary system is to consider a set $\mathcal{X} \subset\mathbb{R}^{nM \times 1}$ 
where $M$ represents the number of targets, and a system for multiple targets:
\begin{align}
    \dot{\mathbf{x}}(t) = I(\mathbf{x}(t), t), \quad \mathbf{z}(t) = J(\mathbf{x}(t), t)
\end{align}
with a super state vector $\mathbf{x}(t) \in \mathcal{X}$ 
representing the states of all $M$ targets, and an output vector $\mathbf{z}(t) \in \mathbb{R}^{M\times 1}$ corresponding to the measurements of these targets. 
\begin{definition}(Observability for Multiple Targets in arbitrary linear or non-linear systems from measurements) \label{def:def2}The system is observable on $\mathcal{X}$ if for all pairs of super state vectors $(\mathbf{x}_m, \mathbf{x}_n) \in \mathcal{X} \times \mathcal{X}$,
\begin{align}
  & J(\mathbf{x}(t; t_{i}, \mathbf{x}_m), t) = J(\mathbf{x}(t; t_{i}, \mathbf{x}_n), t) \quad \forall t \geq t_{i} \\ \nonumber
   & \Rightarrow \mathbf{x}(t; t_{i}, \mathbf{x}_m) = \mathbf{x}(t; t_{i}, \mathbf{x}_n)
\end{align}
where $\mathbf{x}(t; t_{i}, \mathbf{x}_m)$ and $\mathbf{x}(t; t_{i}, \mathbf{x}_n)$ represent the solutions at time $t$ of the system through initial super state vectors $\mathbf{x}_m$ and $\mathbf{x}_n$ at time $t_{i}$, respectively.
\end{definition}
In this context, observability ensures that the measurements associated with different targets are uniquely identifiable, leading to the correct estimation of the state vector $\mathbf{x}(t)$ for the multiple-target system.
\section{Observability Conditions}\label{sec:sec2}
This section outlines the observability conditions for multiple targets when bearing measurements are available. We derived the NECNDSUF condition for observability on the bearing angles of different targets from bearing measurements using Def. \eqref{def:def1}.\\
Let,
$\theta_i(t)$ be the bearing of $i^{th}$ target at time $t$ shown in (Fig.\ref{Fig:Fig 1}) (with the $i^{th}$ target moving from C at time $t$ to D at time $t+1$ and the observer moving from A at time $t$ to B at time $t+1$) given by,
\begin{align}
    {\theta}_i(t)= \tan^{-1} \left( \frac{{x_i}(t)}{{y_i}(t)} \right).
\end{align}
\begin{figure}
    \centering
    \includegraphics[width=0.8\linewidth]{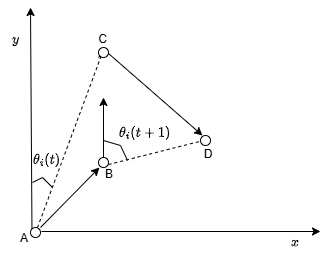}
    \caption{Bearing only tracking}
    \label{Fig:Fig 1}
\end{figure}
{$\mathbf{x}_i(t) \in \mathbb{R}^{1 \times {2}}$} for $i$-th target be defined as, 
\begin{align} \label{eqn. 14}
    \mathbf{x}_i(t) = \begin{bmatrix}
    x_i(t) & y_i(t)
    \end{bmatrix}
\end{align}
Define the super state vector, $\mathbf{x}(t) \in \mathbb{R}^{2M \times 1}$ as $\mathbf{x}(t) = \begin{bmatrix} \mathbf{x}_1(t) & \mathbf{x}_2(t) & \hdots & \mathbf{x}_M(t) \end{bmatrix}^\top $.
Now, we can write in polynomial form:
\begin{align} \label{eqn 22}
    \mathbf{x}_i^T(t) = \sum_{k=0}^{p_i} \frac{\mathbf{{x}_i^{(k)}}^T(t_{i})}{k!} (t - t_{i})^k 
\end{align}
\begin{align} \label{eqn 23}
    = \sum_{k=0}^{p_i} \mathbf{a}^i_k (t - t_{i})^k = \mathbf{A_i} \mathbf{t_i}.
\end{align}
where, $\mathbf{x}_i^{(k)}(t_{i})$ is the $k$-th time derivative of relative position vector $\textbf{x}_i(t_{i})$, $p_i=\max\{N_i,N_{OB}\}$ and $\mathbf{a}^i_k = \frac{\mathbf{{x}_i^{(k)}}^T(t_{i})}{k!}$. Here,
\(\mathbf{A_i} = [\mathbf{a}^i_{0} \ldots \mathbf{a}^i_{p_i}]\) is a \(2(p_i + 1) \times (p_i + 1)\) matrix of coefficients independent of \(t\), 
\(\mathbf{t_i} = \begin{bmatrix}
    1 & (t - t_{i}) & \ldots & (t - t_{i})^{p_i}
\end{bmatrix}^\top\)
and 
    $\mathbf{x}(t) = \mathbf{A}\mathbf{t},$  
where, $\mathbf{A} \in \mathbb{R}^{2s \times s} $ and $ \mathbf{t} \in \mathbb{R}^{s \times 1}$ be defined as,
$\mathbf{A} = \mathrm{diag}({\mathbf{A_1}, \mathbf{A_2}, \ldots, \mathbf{A_M}}) \quad \text{and} \quad \mathbf{t} = \begin{bmatrix}
    \mathbf{t_1} & \mathbf{t_2} & \hdots & \mathbf{t_M}
\end{bmatrix}^\top.$
Let us define the Pseudo Linearized measurement for the $i^{th}$ target as,
\[
\mathbf{z}_i(t) = \begin{bmatrix}
c\o s(\theta_i(t)) & -sin(\theta_i(t)) & \mathbf{0}_{1 \times 2p_i}
\end{bmatrix}
\mathbf{x}_i^\top(t)
\]
\[
\quad = x_i c\o s(\theta_i(t)) - y_i sin(\theta_i(t)).
\]
The Pseudo Measurement vector for $M$ targets is $\mathbf{z}(t) = \mathbf{C}(t) \mathbf{x}(t)$,
where $\mathbf{C}(t) \in \mathbb{R}^{M \times 2s}$ be defined as,
\begin{align}\label{eqn.17}
    \mathbf{C}(t) = \begin{bmatrix}
    \mathbf{c_1(t)} & \mathbf{0}_{1\times 2(p_2+1)} & \cdots & \mathbf{0}_{1\times 2(p_M+1)} \\
    \mathbf{0}_{1\times 2(p_1+1)} & \mathbf{c_2(t)} & \cdots & \mathbf{0}_{1\times 2(p_M+1)} \\
    \vdots & \vdots & \ddots & \vdots \\
    \mathbf{0}_{1\times 2(p_1+1)} & \mathbf{0}_{1\times2(p_2+1)} & \cdots & \mathbf{c_M(t)}
\end{bmatrix},
\end{align}
and,
\begin{align}\label{eqn: 18}
   \mathbf{c_i}(t) = \begin{bmatrix}
    c\o s(\theta_i(t)) & -sin(\theta_i(t)) & \mathbf{0}_{1 \times 2p_i}
    \end{bmatrix}. 
\end{align}
Let, $\tilde{\mathbf{\Phi}}_i(t, t_{i}) \in \mathbb{R}^{2(p_i + 1) \times 2(p_i + 1)}$ be defined below;
\begin{align} \label{eqn: 19}
   \tilde{\mathbf{\Phi}}_i(t, t_{i}) = \begin{bmatrix}
    1 & 0 & (t - t_{i}) & 0 &\cdots &(t - t_{i})^{p_i}  \\
    0 & 1 & 0 & (t - t_{i}) &\cdots &(t - t_{i})^{p_i}\\
    \vdots & \vdots & \vdots & \vdots & \ddots & \vdots\\
    0 & 0 & 0 & 0 &\cdots &(t - t_{i})^{p_i}
    \end{bmatrix}.
\end{align}
Now for any unforced system (i.e., $\tilde{R} = 0$) \eqref{eqn. 3},
\begin{align}\label{eqn. 18}
    \mathbf{C}(t) \tilde{\mathbf{\Phi}}(t, t_{i}) \mathbf{x}(t_{i}) = \mathbf{z}(t), \quad\text{i.e.,}\quad\mathbf{x}(t) = \tilde{\mathbf{\Phi}}(t, t_{i}) \mathbf{x}(t_{i}).
\end{align}
Here, $\tilde{\mathbf{\Phi}}(t, t_{i})\in \mathbb{R}^{2s \times 2s}$ is a block diagonal matrix with each diagonal block being $\tilde{\mathbf{\Phi}}_i(t, t_{i})$.
Now from \eqref{eqn 22},\eqref{eqn 23},\eqref{eqn: 19} and \eqref{eqn. 18} we can write for $i$-th target,
\[
\mathbf{x}_i^\top(t) = \mathbf{x}_i^\top(t_{i}) + \mathbf{x}_i^{(1)\top}(t_{i}) (t - t_{i})+ \hdots +\frac{\mathbf{x}_i^{(p_i)\top}(t_{i})}{p_i!} (t - t_{i})^{p_i} 
\]
\[
= \mathbf{A_i} \mathbf{t_i} = \tilde{\mathbf{\Phi}}_i(t,t_{i}) \mathbf{x}_i^\top(t_{i}) 
\quad \quad \forall i \in \mathcal{M}.
\]

Consider contrapose of the observability condition defined in Definition (\ref{def:def1}) i.e., for the system to be observable on \([t_{i}, t_{f}]\) for any $\mathbf{\tilde{x}} \in\mathbb{R}^{2s \times 1}$,
\begin{align}\label{eqn:: 21}
    \{ \forall t \in [t_{i}, t_{f}],|| \mathbf{C}(t) \tilde{\mathbf{\Phi}}(t, t_{i}) \mathbf{\tilde{x}||} = 0 \} \implies \{ \mathbf{\tilde{x}} = 0 \}.
\end{align}
Let, $\mathbf{\tilde{x}}=\mathbf{x}(t_{i})$ and the initial state for the $i$-th target at time $t_{i}$ be defined as \eqref{eqn. 14}, $\mathbf{x_i}(t_{i}) = \begin{bmatrix} x_i(t_{i}) & y_i(t_{i}) & \hdots & x^{(p_i)}_{i}(t_{i}) & y^{(p_i)}_{i}(t_{i}) \end{bmatrix} $. Now, for $||\mathbf{C}(t) \tilde{\mathbf{\Phi}}(t, t_{i}) \mathbf{\tilde{x}}||$ to be equal to zero, $\mathbf{c}_i(t) \tilde{\mathbf{\Phi}}_i(t, t_{i}) \mathbf{x}_i^T(t_{i})$ should be equal to zero $\forall i\in \mathcal{M}$ and from \eqref{eqn. 14}, \eqref{eqn: 18} and \eqref{eqn: 19}, $\mathbf{c}_i(t) \tilde{\mathbf{\Phi}}_i(t, t_{i}) \mathbf{x}_i^T(t_{i}) = 0$ implies,
\begin{align}\label{eqn: eqn:27}
  c\o s&{\theta_i}(t)x_i(t_{i}) +c\o s{\theta_i}(t)x_i^{(1)}(t_{i}) (t - t_{i}) + \hdots + \nonumber \\ &c\o s{\theta_i}(t)\frac{x_i^{(p_i)}(t_{i})}{p_i!} (t - t_{i})^{p_i} -
  sin{\theta_i}(t)y_i(t_{i}) - \nonumber \\ & sin{\theta_i}(t)y_i^{(1)}(t_{i}) (t - t_{i}) - \hdots -sin{\theta_i}(t)\frac{y_i^{(p_i)}(t_{i})}{p_i!} (t - t_{i})^{p_i} \nonumber \\ & = 0 \quad \quad \quad \quad \quad \quad \quad \quad \quad \quad\quad \quad\quad \quad \forall i \in \mathcal{M}, 
\end{align}

where $\mathbf{x}_i^{(k)}(t_{i})$ and $\mathbf{y}_i^{(k)}(t_{i})$ ($\forall i \in \mathcal{M}$ and $\forall k \in \{1,2,\hdots p_i\}$) are the unknown variables.
To ensure the observability condition \eqref{eqn:: 21} the matrix \(\textbf{P}(t) \) given below needs to be full rank (\textit{i.e.}, it should have rank 2).
\[
\mathbf{P}(t) = \begin{pmatrix}
c\o s(\theta_1(t)) & -sin(\theta_1(t)) \\
c\o s(\theta_2(t)) & -sin(\theta_2(t)) \\
c\o s(\theta_3(t)) & -sin(\theta_3(t)) \\
\vdots & \vdots \\
c\o s(\theta_M(t)) & -sin(\theta_M(t))
\end{pmatrix}.
\]
Let us consider the sub matrix \( \textbf{M} \) with two rows:
\[
\mathbf{M}(t) = \begin{pmatrix}
c\o s(\theta_i(t)) & -sin(\theta_i(t)) \\
c\o s(\theta_j(t)) & -sin(\theta_j(t))
\end{pmatrix},
\]
For \( \textbf{M}(t) \) to be full rank, its determinant must be non-zero. 
This shows that:
\begin{align}
    \theta_j(t) - \theta_i(t) \neq k\pi \quad \text{for} \quad k \in \mathbb{Z}.
\end{align}
In other words, \(\theta_j(t)\) should not be equal to \(\theta_i(t) + k\pi\) for any integer \( k \). 

This condition has to hold $\forall i \in \mathcal{M}, j \in \mathcal{M}$, $i \neq j$ and   
$\forall k \in \mathbb{Z}.$
Thus, the general condition for the matrix $\textbf{P}(t)$ to be full rank is that all angles $\theta_i(t), i \in \mathcal{M},$ should be distinct modulo $\pi$, where odd $k$ implies observer in between targets and even $k$ implies targets on the same side of the observer.

\begin{proposition}
    The system of $M$ targets and an observer is observable 
    if none of the targets in their trajectory is positioned along the line joining the observer and another target for all time instants $t$ in the interval \([t_{i}, t_{f}]\). In other words, the bearing angles of each target should be distinct modulo $\pi$ for all $t \in [t_{i},t_{f}]$.
\end{proposition}
\begin{remark}
    Out of $M$ targets, if any two targets are collinear with the observer for any $ t \in  [t_{i},t_{f}]$, the system remain unobservable.
    Thus, the system will remain unobservable if any two or more targets are collinear with the observer \textit{i.e.}, when two or more targets lie on the line in sky blue (see Fig.~\ref{fig:fig 2}).
\end{remark}
\begin{remark}
    In 3D, where bearing and elevation angles are given, either the bearing angles or the elevation angles should be distinct modulo $\pi$ for all $t \in [t_{i},t_{f}]$ for observability.
\end{remark}
\begin{remark}
    The condition for the observability of multiple targets, derived using Def. \eqref{def:def1} and based on bearing measurements, is independent of the order of dynamics of both the observer and the targets.
\end{remark}

\begin{figure}[h!]
    \centering
    \includegraphics[width=0.8\linewidth]{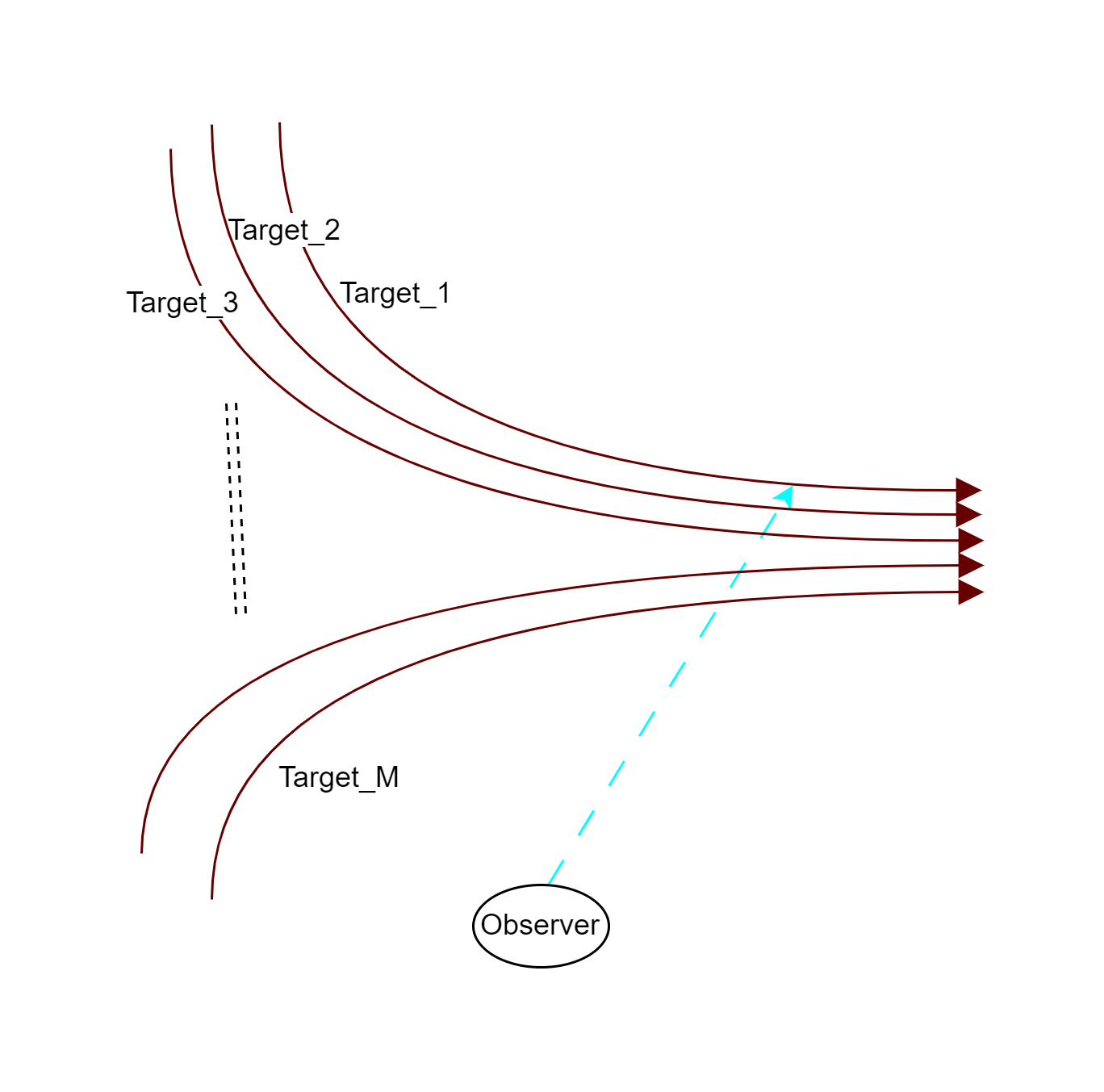}
    \caption{Multi-Target Scenario}
    \label{fig:fig 2}
\end{figure}
\section{Observability through Trajectory Ambiguity}\label{sec:sec3}
Another notion of observability in the multi-target scenario is to ensure that the trajectory of one target is not a compatible trajectory (\textit{i.e.}, the trajectory with identical measurement) of any other target. In such a case, shrinking the compatible trajectory of one target (which is the true trajectory of another target) to its true trajectory to ensure individual observability may become impossible and can hinder observability of the system of $M$ targets, as it would prevent unique state estimation for each individual target.
\begin{definition}(Trajectory Ambiguity)
    Let \( \mathcal{M} = \{1, 2, \dots, M\} \) represent the set of $M$ targets. Two or more targets 
    are said to have ambiguous trajectories if they have the same measurement histories $\forall t \in [t_{i},t_{f}]$, which can be represented as:
\begin{align}
    z_i(t) = z_j(t) = \hdots = z_k(t)\quad \forall i, j,\hdots k \in \mathcal{M}, \nonumber \\ \ i \neq j \neq \hdots \neq k \quad \text{and} \quad \forall t \in [t_{i},t_{f}], 
\end{align}
where, \( z_i(t) \), \( z_j(t) \) and \( z_k(t) \) represent the measurements of the \( i \)-th, \( j \)-th and \( k \)-th targets.
\end{definition}
For simplicity, we have derived the NECNDSUF conditions for ambiguous trajectories with respect to Doppler frequencies, bearing angles, and the combination of Doppler frequencies and bearing angles, considering two targets out of the $M$ targets.
\subsection{Necessary condition for ambiguity in trajectories with Doppler frequencies}
\begin{proposition}
    Multiple targets will have ambiguous trajectories with respect to Doppler frequencies iff they satisfy the following condition:
    \begin{align}\label{eqn 39}
        \tilde{\mathbf{s}}_i(t) - \tilde{\mathbf{s}}_j(t) = (\mathbf{W}(t) - \mathbf{I}) \mathbf{s}_j(t), \quad \forall t \in [t_{i}, t_{f}],
    \end{align}
    for any \(i \in M\), \(j \in M\), and \(i \neq j\). Here, \(\tilde{\mathbf{s}}_i(t)\) represents the position vector of the \(i\)-th target at time \(t\), and \(\tilde{\mathbf{s}}_j(t)\) represents the position vector of the \(j\)-th target at time \(t\). The relative position vector of the \(i\)-th target at time \(t\) is given by: $\mathbf{s}_i(t) = \tilde{\mathbf{s}}_i(t) - \mathbf{s}_\text{OB}(t),$
    and the relative position vector of the \(j\)-th target at time \(t\) is given by: $\mathbf{s}_j(t) = \tilde{\mathbf{s}}_j(t) - \mathbf{s}_\text{OB}(t).$
    The vector \(\mathbf{s}_\text{OB}(t)\) denotes the position of the observer (e.g., radar, sensor) at time \(t\). $\mathbf{W}(t) = \mathbf{U(t)} \left[ l' +  \frac{b' + c(1 - l')(t - t_{i})}{s_j(t)} \right]$ (is derived and explained in the proof below) is a transformation matrix that varies with time \(t\) and represents some dynamic transformation, while \(\mathbf{I}\) is the identity matrix.\\
\end{proposition}
\begin{proof}
    Let the Doppler frequency measurement history for $i^{th}$ and $j^{th}$ targets  \(f_{D,i}(t)\) and \(f_{D,j}(t)\), $\forall t \in [t_{i},t_{f}]$, be the same,
\begin{align} \label{eqn 46}
    f_{i,0}\left(1 - \frac{\dot{s}_i(t)}{c}\right) = f_{j,0} 
\left(1 - \frac{\dot{s}_j(t)}{c}\right) \quad \forall t \in [t_{i},t_{f}]
\end{align} 
for any \( i \in \mathcal{M} \), \( j \in \mathcal{M} \) and $i\neq j$ where $f_{i,0}$ and $f_{j,0}$ are the actual radiated tonals from $i^{th}$ and $j^{th}$ targets, $\dot{s}_i(t)= \frac{\mathbf{\dot{s_i}}^T(t)\mathbf{s_i}(t)}{{s_i}(t)} $ and $\dot{s}_j(t)= \frac{\mathbf{\dot{s_j}}^T(t)\mathbf{s_j}(t)}{{s_j}(t)} $ are the relative velocities of $i^{th}$ and $j^{th}$ targets, $\mathbf{s_i}(t)$ and $\mathbf{s_j}(t)$ are the relative positional vectors of  $i^{th}$ and $j^{th}$ targets.\\ 
Now, integrating \eqref{eqn 46} from $t_{i}$ to $t$ and rearranging we obtain,
\begin{align}\label{eqn 47}
    s_i(t) = l's_j(t) + b' + c(1-l')(t-t_{i})
\end{align}
where,
$l' = \frac{f_{j,0}}{f_{i,0}}$$\quad \& $$\quad b' = s_{i,0} - l's_{j,0}$ \\
$s_{i,0}$ and $s_{j,0}$ are two initial positions of $i^{th}$ and $j^{th}$ targets.
Thus, if $\hat{\mathbf{u}}_{s_i(t)}$ and $\hat{\mathbf{u}}_{s_j(t)}$ are the unit vectors along the directions of $\textbf{s}_i(t)$ and $\textbf{s}_j(t)$ respectively, we can write, based on \eqref{eqn 47}, and considering the transformation of unit vectors $\hat{\mathbf{u}}_{s_i(t)}$ and $\hat{\mathbf{u}}_{s_j(t)}$ by an arbitrary unitary transformation $\mathbf{U(t)}$ for all $t \in [t_{i}, t_{f}]$:
\begin{align*}
\mathbf{s}_i(t) &= s_i(t)\hat{\mathbf{u}}_{s_i(t)} \\
                &= [l' s_j(t) + b' + c(1 - l')(t - t_{i})]\hat{\mathbf{u}}_{s_i(t)} \quad \\
                &= \mathbf{U(t)}[l' s_j(t) + b' + c(1 - l')(t - t_{i})]\hat{\mathbf{u}}_{s_j(t)} \\
                &= \mathbf{U(t)}[l' +\frac{b' + c(1 - l')(t - t_{i})}{s_j(t)}]\mathbf{s}_j(t)
\end{align*}
\begin{align*}
    \Rightarrow  
\mathbf{\tilde{s}}_i(t) -\mathbf{s}_{OB}(t) = \mathbf{U(t)} 
\left[ 
l' +  \frac{b' + c(1 - l')(t - t_{i})}
{s_j(t)} \right]
(\mathbf{\tilde{s}}_j(t) -\mathbf{s}_{OB}(t))
\end{align*}
\begin{align} \label{eqn. 28:}
    \quad \Rightarrow \mathbf{\tilde{s}}_i(t) = \mathbf{W}(t) \mathbf{\tilde{s}}_j(t) +\mathbf{s}_{OB}(t) (\mathbf{I}-\mathbf{W}(t))
\end{align}
where $\mathbf{W}(t) = \mathbf{U(t)} 
\left[ 
l' +  \frac{b' + c(1 - l')(t - t_{i})}
{s_j(t)} \right] $  and the function $\mathbf{U(t)}$ is designed to offer a unitary transformation, which ensures that one unit vector is mapped or transformed into another unit vector. This transformation preserves the length and orthogonality of the unit vectors, maintaining their unitary properties. Now \eqref{eqn. 28:} can be written as:
\[
\mathbf{\tilde{s}}_i(t) - \mathbf{\tilde{s}}_j(t) = (\mathbf{W}(t)-\mathbf{I}) \mathbf{s}_j(t)
\]

\end{proof} 

\begin{remark}
   The necessary condition for Ambiguity in Trajectories with respect to Doppler frequencies derived previously describes the relationship between the position vectors of the targets $i$-th and $j$-th. It indicates that if the relative positions of the two targets are related by a transformation matrix $\mathbf{W}(t)$, then their trajectories will be ambiguous; however, it is not sufficient. For sufficiency, this condition alone does not guarantee that the Doppler frequency measurements will be identical,it only provides a constraint on how their relative positions are related. Along with the previously established necessary condition, new sufficient conditions are:
    \begin{enumerate}
        \item \( f_{i,0} = f_{j,0} \) \quad \text{(same radiated tonal)},
        \item \( \mathbf{W}(t) = I \) \quad \text{(no relative transformation of velocities)}
        \item $\mathbf{s}_i(t) = {\mathbf{s}}_j(t).$
    \end{enumerate}
    \begin{proof}
        See Appendix~\ref{necc_suff_dopp}.
    \end{proof}

\end{remark}

\subsection{NECNDSUF condition for ambiguity in trajectories with bearing angles}
Let the bearing measurement histories for $i^{th}$ and $j^{th}$ targets be the same, \textit{i.e.}, $ \mathbf{s}_i(t) = \alpha'(t) \mathbf{s}_j(t) \quad \forall t \in [t_{i},t_{f}] $
for any \( i \in \mathcal{M} \), \( j \in \mathcal{M} \), $i\neq j$ and $\alpha'(t)$ being any scalar constant at time t.
Then, the NECNDSUF condition for ambiguity in trajectories between the $i$-th and the $j$-th targets is:
\begin{align}
    \quad  \mathbf{\tilde{s}}_i(t) - \mathbf{\tilde{s}}_j(t) = (\alpha'(t)-1) \mathbf{s}_j(t) \quad \forall t \in [t_{i},t_{f}].
\end{align}

\subsection{NECNDSUF condition for ambiguity in trajectories with Doppler and bearing measurements} 

\begin{proposition}
The NECNDSUF condition for an ambiguous trajectory, derived from measurements with Doppler frequency and bearing, implies that \(\mathbf{s}_j(t)\) must be an eigenvector of \(\mathbf{W}(t)\) corresponding to the eigenvalue \(\alpha(t)\). Where, \(\mathbf{\tilde{s}}_j(t)\) is the position vector of the \(j\)-th target at time \(t\).  
\(\mathbf{W}(t)\) is a transformation matrix that varies with time \(t\) and represents some dynamic transformation.  
\(\alpha(t)\) is a unit scalar eigenvalue associated with the eigenvector \(\mathbf{s}_j(t)\) at time \(t\).  
\end{proposition}

\begin{proof}
    See Appendix~\ref{necc_suff_bearing_doppler}.
\end{proof}

\section{Conclusion}
In this paper, the conditions for observability have been derived for targets and observers with different higher-order dynamics, assuming a constant course. These conditions can be extended to scenarios where the target or observer maneuvers (i.e., the course is not constant). Additionally, the corresponding observability conditions can be formulated for cases where the number of measurements is constrained.
\vspace{-8pt}
\appendix
\section{Proof of \eqref{eqn. 3} }
\label{observability proof}

The differential equation is:
\begin{align} \label{eqn. 42}
    \frac{d\mathbf{x}(t)}{dt} - \mathbf{E}(t)\mathbf{x}(t) = \mathbf{F}(t)\mathbf{u}(t).
\end{align}

Let, $\tilde{\boldsymbol{\Phi}}(t, t_{i})^{-1}$ be the integrating factor, where $\tilde{\boldsymbol{\Phi}}(t, t_{i})$ is defined as,
$\tilde{\boldsymbol{\Phi}}(t, t_{i}) = \exp\left( \int_{t_{i}}^t \mathbf{E}(\tau) \, d\tau \right),$
and \(\tilde{\boldsymbol{\Phi}}(t, t_{i})\) is the \textbf{"state transition matrix"}, satisfying:
\[
\frac{\partial}{\partial t}\tilde{\boldsymbol{\Phi}}(t, t_{i}) = \mathbf{E}(t)\tilde{\boldsymbol{\Phi}}(t, t_{i}), \quad \quad \tilde{\mathbf{\Phi}}(t_{i},t_{i}) = \mathbf{I}_{nM \times nM}.
\]
Multiply both sides of \eqref{eqn. 42} by the integrating factor \(\tilde{\boldsymbol{\Phi}}(t, t_{i})^{-1}\), we can obtain:
\[
\tilde{\boldsymbol{\Phi}}(t, t_{i})^{-1} \frac{d\mathbf{x}(t)}{dt} - \tilde{\boldsymbol{\Phi}}(t, t_{i})^{-1} \mathbf{E}(t)\mathbf{x}(t) = \tilde{\boldsymbol{\Phi}}(t, t_{i})^{-1} \mathbf{F}(t)\mathbf{u}(t).
\]
\[
\Rightarrow \frac{d}{dt} \left( \tilde{\boldsymbol{\Phi}}(t, t_{i})^{-1}\mathbf{x}(t) \right) = \tilde{\boldsymbol{\Phi}}(t, t_{i})^{-1} \mathbf{F}(t)\mathbf{u}(t).
\]
Integrate both sides from \(t_{i}\) to \(t\) and rearrange:

\[
\mathbf{x}(t) = \tilde{\boldsymbol{\Phi}}(t, t_{i})\mathbf{x}(t_{i}) + \int_{t_{i}}^t \tilde{\boldsymbol{\Phi}}(t, \tau)\mathbf{F}(\tau)\mathbf{u}(\tau) \, d\tau.
\]

\vspace{-6pt}
\section{Proof of Nec and Suff condition for Doppler ambiguity} 
\label{necc_suff_dopp}
The proof proceeds by assuming that the necessary condition holds and then establishing that the condition for equal Doppler measurements provides the sufficient condition.
We are given the necessary condition in \eqref{eqn 39} that the position vectors 
\(\tilde{\mathbf{s}}_i(t)\) and \(\tilde{\mathbf{s}}_j(t)\) of targets \(i\) and \(j\) must satisfy the equation: $\tilde{\mathbf{s}}_i(t) - \tilde{\mathbf{s}}_j(t) = (\mathbf{W}(t) - \mathbf{I})s_j(t)$,for all $t \in [t_{i}, t_{f}]$.

Differentiating both sides with respect to time \(t\), we get, for all $t \in [t_{i}, t_{f}]$:

\begin{align}\label{eqn 48}
    \dot{\mathbf{s}}_i(t)  = \dot{\mathbf{s}}_j(t) + \dot{\mathbf{W}}(t)\mathbf{s}_j(t) + (\mathbf{W}(t) - I)\dot{\mathbf{s}}_j(t)
\end{align}

Now, \eqref{eqn 46} can be rewritten as:
\begin{align} \label{eqn 49}
   f_{i,0} \frac{\dot{\mathbf{s}}_i^\top(t)\mathbf{s}_i(t)}{s_i(t)} - f_{j,0} \frac{\dot{\mathbf{s}}_j^\top(t)\mathbf{s}_j(t)}{s_j(t)} = c \left(f_{i,0} - f_{j,0}\right) 
\end{align}
Now, let’s substitute \eqref{eqn 48} in \eqref{eqn 49} and simplifying yields;

\begin{align*}
    f_{i,0}\frac{\dot{\mathbf{s}}_j^\top(t)\mathbf{s}_i(t) + \mathbf{s}_j^\top(t)\dot{\mathbf{W}}^\top(t)\mathbf{s}_i(t) + \dot{\mathbf{s}}_j^\top(t)(\mathbf{W}(t) - \mathbf{I})^\top\mathbf{s}_i(t)}{s_i(t)}\\=
    f_{j,0} \frac{\dot{\mathbf{s}}_j^\top(t)\mathbf{s}_j(t)}{s_j(t)}.
\end{align*}
Thus, the sufficient condition entails:
$f_{i,0} = f_{j,0} , \mathbf{W}(t) = \mathbf{I} $ and $\mathbf{s}_i(t) = {\mathbf{r}}_j(t)$ for all $t \in [t_{i}, t_{f}]$.

\vspace{-8pt}
\section{Proof of Nec and Suff condition for Bearing and Doppler ambiguity} 
\label{necc_suff_bearing_doppler}
The necessary condition for Ambiguous trajectory derived from Doppler frequency measurement is: $\mathbf{\tilde{s}}_i(t) - \mathbf{\tilde{s}}_j(t) = (\mathbf{W}(t)-\mathbf{I}) \mathbf{s}_j(t)$.

Where, $\mathbf{\tilde{s}}_i(t)$, $\mathbf{\tilde{s}}_j(t)$, $\mathbf{s}_j(t)$ , $\mathbf{s}_{OB}(t)$, $\mathbf{W}(t)$, $\mathbf{I}$ are as explained previously.
The NECNDSUF condition for ambiguous trajectory derived from only bearing measurement is: $\tilde{\mathbf{s}}_i(t) - \tilde{\mathbf{s}}_j(t) = (\alpha'(t) - 1) \mathbf{s}_j(t)$.
Where  $\tilde{\mathbf{s}}_i(t)$, $\tilde{\mathbf{s}}_j(t) $, $\mathbf{s}_j(t)$, $\alpha'(t)$ are as explained previously.
Now from these two conditions, by equating the right-hand sides of the above equations, we obtain the necessary condition as: $(\mathbf{W}(t) - I) \mathbf{s}_j(t) = (\alpha'(t) - 1) \mathbf{s}_j(t)$
which simplifies to:
\begin{align}\label{eqn 50}
    \mathbf{W}(t) \mathbf{s}_j(t) = \alpha'(t) \mathbf{s}_j(t)   
\end{align}
This implies that  $\mathbf{s}_j(t)$ is an eigenvector of the matrix $\mathbf{W}(t)$ corresponding to the eigenvalue $ \alpha'(t)$. From Appendix \ref{necc_suff_dopp} and \eqref{eqn 50}, the NECNDSUF condition for Doppler bearing tracking is $\alpha'(t) = 1$

\section*{Conflict of Interest}
The authors declare that they have no conflict of interest.

\bibliographystyle{unsrt}
\bibliography{Reference}
\end{document}